\documentclass[times,authoryear]{elsarticle}
\usepackage{amsmath}
\usepackage{amsthm}
\usepackage{amssymb}
\usepackage{amsfonts}

\newtheorem{theorem}{Theorem}

\newtheorem{corollary}[theorem]{Corollary}
\newtheorem{conjecture}[theorem]{Conjecture}

\newtheorem{definition}[theorem]{Definition}

\usepackage{tikz}
\usetikzlibrary{arrows}
\usepackage{caption}
\usepackage{mathtools}
\usepackage[OT2,T1]{fontenc}
\usepackage{microtype}
\renewcommand{\rmdefault}{ptm}
\usepackage[noend]{algpseudocode}
\usepackage[pdftex,pagebackref=false,hidelinks]{hyperref}

\DeclareMathOperator{\PAF}{PAF}
\DeclareMathOperator{\rowsum}{rowsum}
\DeclareMathOperator{\DFT}{DFT}
\DeclareMathOperator{\PSD}{PSD}
\DeclareMathOperator{\Aut}{Aut}

\newcommand{\N}{\mathbb{N}}
\newcommand{\Z}{\mathbb{Z}}
\newcommand{\E}{\text{E}}

\let\brace\undef
\let\brack\undef
\DeclarePairedDelimiter{\abs}{\lvert}{\rvert}

\DeclarePairedDelimiter{\floor}{\lfloor}{\rfloor}
\DeclarePairedDelimiter{\brack}{\lbrack}{\rbrack}
\DeclarePairedDelimiter{\brace}{\lbrace}{\rbrace}

\newcommand{\MC}{\textsc{MathCheck}}
\renewcommand{\today}{March 15, 2019}

\newcommand\cyr
{%
\renewcommand\rmdefault{wncyr}%
\renewcommand\sfdefault{wncyss}%
\renewcommand\encodingdefault{OT2}%
\normalfont%
\selectfont%
}

\makeatletter
\def\ps@pprintTitle{%
     \let\@oddhead\@empty
     \let\@evenhead\@empty
     \def\@oddfoot{\footnotesize\itshape
       To appear in the \ifx\@journal\@empty Journal of Symbolic Computation
       \else\@journal\fi\hfill\today}%
     \let\@evenfoot\@oddfoot}
\makeatother

\newcommand{\shuffle}{\mathbin{\text{\cyr x}}}

\begin{document}

\begin{frontmatter}

\title{Applying Computer Algebra Systems with SAT Solvers
to the Williamson Conjecture}

\author[uw]{Curtis Bright}
\author[wlu]{Ilias Kotsireas}
\author[uw]{Vijay Ganesh}
\address[uw]{University of Waterloo}
\address[wlu]{Wilfrid Laurier University}

\begin{abstract}
We employ tools from the fields of symbolic computation and satisfiability
checking---namely, computer algebra systems and SAT solvers---to study the
Williamson conjecture from combinatorial design theory and increase the
bounds to which Williamson matrices have been enumerated.
In particular, we completely enumerate all Williamson matrices of
even order up to and including $70$ which gives
us deeper insight into the behaviour and distribution of Williamson matrices.
We find that, in contrast to the case when the order is odd, Williamson
matrices of even order are quite plentiful and exist in every even order
up to and including $70$.
As a consequence of this and a new construction for 8-Williamson matrices
we construct 8-Williamson matrices in all odd orders up to and including $35$.
We additionally enumerate all Williamson matrices whose orders are divisible
by $3$ and less than $70$, finding one previously unknown set of
Williamson matrices of order $63$.
\end{abstract}


\end{frontmatter}

\section{Introduction}\label{sec:introduction}

In recent years SAT solvers have been used to solve or make progress on
mathematical conjectures which have otherwise
resisted solution from some of the world's best mathematicians.
Some prominent problems which fit into this trend
include the Erd\H{o}s discrepancy conjecture, which was open for
80 years and had a special case solved by~\cite{DBLP:conf/sat/KonevL14};
the Ruskey--Savage conjecture, which has been open for 25 years
and had a special case solved by~\cite{DBLP:conf/cade/ZulkoskiGC15};
the Boolean Pythagorean triples problem, which was open for 30 years
and solved by~\cite{heule2016solving}; and the
determination of the fifth Schur number, which was open for 100 years
and solved by~\cite{DBLP:journals/corr/abs-1711-08076}.
Although these are problems
which arise in completely separate fields
and have no obvious connection to propositional
satisfiability checking, nevertheless SAT solvers were found
to be extremely effective at pushing the state-of-the-art and
sometimes absolutely crucial in the problem's ultimate solution.

In this paper we apply a SAT solver to the Williamson conjecture from
combinatorial design theory.
Our work is similar in spirit to the aforementioned works but
we would like to highlight two main differences.
Firstly, we employ an approach
inspired by SMT (SAT modulo theories) solvers
and use a SAT solver 
that is able to learn conflict clauses through a piece of code
specifically tailored to the problem domain.
This code encodes domain-specific knowledge that an
off-the-shelf SAT solver would otherwise not be able to exploit.  This
framework is not limited to any specific domain;
any external library or function can be used
as long as it is callable by the SAT solver.  As we will see
in Section~\ref{sec:sat+cas}, the clauses that
are learned in this fashion can enormously cut down the search space as well
as the solver's runtime.

Secondly, similar in style to~\citep{DBLP:conf/cade/ZulkoskiGC15}
we incorporate functionality from computer algebra systems to
increase the efficiency of the search in what we
call the ``SAT+CAS'' paradigm. 
This approach of combining computer algebra systems with SAT
or SMT solvers was also independently proposed
at the conference ISSAC by~\cite{DBLP:conf/issac/Abraham15}.
More recently, it has been argued by the SC$^2$ project~\citep{sc2} that
the fields of satisfiability checking and symbolic computation
are complementary and combining the tools of both fields
(i.e., SAT solvers and computer algebra systems)
in the right way can solve problems more efficiently than
could be done by applying the tools of either field in isolation,
and our work provides evidence for this view.

We describe the Williamson conjecture, its history, and state the necessary
properties of Williamson matrices that we require in Section~\ref{sec:willconj}.
In particular, we derive a new version of Williamson's product theorem that applies
to Williamson matrices of even order (Theorem~\ref{thm:willprodeven}).
We give an overview of the
SAT+CAS paradigm in Section~\ref{sec:sat+cas}, describe our SAT+CAS method
in Section~\ref{sec:sat+casmethod}, and give a summary
of our results in Section~\ref{sec:results}.
The present work is an extension of our
previous work~\citep{bright2017sat+} that enumerated Williamson matrices
of even order up to order~$64$.  The present work extends this enumeration
to order~$70$ and extends the method to enumerate Williamson matrices
with orders divisible by~$3$.  In doing so, we find a previously
undiscovered set of Williamson matrices of order~$63$, the first new set of Williamson
matrices of odd order discovered since 
one of order $43$ was found
over ten years ago by~\cite{holzmann2008williamson}.
Additionally, we improve our treatment of equivalence checking (see Section~\ref{sec:equivcheck}),
identify a new equivalence operation that applies to Williamson
matrices of even order (see Section~\ref{sec:willequiv}),
derive a new doubling construction for Williamson matrices (Theorem~\ref{thm:willdbl}),
and a new construction for 8-Williamson matrices (Theorem~\ref{thm:8will}).
Using this construction we construct 8-Williamson matrices in all
odd orders $n\leq35$, improving on the result of \cite{kotsireas2009hadamard}
that constructed 8\nobreakdash-Williamson matrices in all odd orders $n\leq29$.
Finally, in Section~\ref{sec:conclusion}
we use our experience developing
systems that combine SAT solvers with computer algebra systems
to give some guidelines about the kind of problems for which an approach is
likely to be effective.

\section{The Williamson conjecture}\label{sec:willconj}

\cite{williamson1944hadamard} introduced the matrices which
now bear his name (see Section~\ref{sec:background})
while developing a method of constructing Hadamard matrices.
Hadamard matrices are square
matrices with $\pm1$ entries and pairwise orthogonal rows;
they have a long history and many applications
such as to error-correcting codes~\citep{bose1959note}.
The \emph{Hadamard conjecture}
states that Hadamard matrices exist for all orders divisible by~$4$.
Williamson's construction has been extensively used to construct Hadamard matrices
in many different orders and the \emph{Williamson conjecture}
states that it can be used to construct a Hadamard matrix of any order divisible by~$4$;
\cite{turyn1972infinite} states it as follows:
\begin{quote}
Only a finite number of Hadamard matrices of Williamson type are known
so far; it has been conjectured that one such exists of any order $4t$.
\end{quote}
Williamson matrices have also found use in digital
communication systems and this motivated mathematicians from NASA's
Jet Propulsion Laboratory to construct Williamson matrices of order~$23$
while developing codes allowing the transmission of
signals over a long range~\citep{baumert1962}.
These Williamson matrices were consequently used to 
construct a Hadamard matrix of order $4\cdot23=92$~\citep{nasablog}.
(In some older works the Hadamard matrix constructed in this way was itself referred to
as a Williamson matrix but we follow modern convention and do not use
this terminology.)
Williamson matrices are also studied for their elegant mathematical properties
and their relationship to other mathematical conjectures~\citep{schmidt1999williamson}.

Although Williamson defined his matrices for both even and odd orders, 
most subsequent work has focused on the odd case.
A complete enumeration of Williamson matrices was completed
for all odd orders up to~$23$ by~\cite{baumert1965hadamard}.
A enumeration in orders~$25$ and~$27$ was completed by~\cite{40002775009}
but this enumeration was later found to be incomplete by~\cite{djokovic1995note},
who gave a complete enumeration in the order~$25$ as well as (in a previous paper)
the orders~$29$ and~$31$~\citep{dokovic1992williamson}.
The orders~$33$ and~$39$ were
claimed to be completely enumerated by~\cite{koukouvinos1988hadamard,koukouvinos1990there}
but these searches were demonstrated to be incomplete when a complete
enumeration of the orders $33$, $35$, and~$39$ was completed
by~\cite{dokovic1993williamson}.
Most recently, all odd orders up to~$59$ were enumerated by~\cite{holzmann2008williamson}
and the order $61$ was enumerated by~\cite{Lang2012}.

Historically, less attention was paid to the even order cases,
although generalizations of Williamson matrices
were explicitly constructed in even orders by~\cite{Wallis1974}
as well as~\cite{agayan1981recurrence}.
Williamson matrices were constructed in all
even orders up to~$22$ by~\cite{kotsireas2006constructions},
up to~$34$ by~\cite{DBLP:conf/casc/BrightGHKNC16},
and up to~$42$ by~\cite{DBLP:journals/jar/ZulkoskiBHKCG17}.
\cite{kotsireas2006constructions} provided a exhaustive search
up to order $18$ but otherwise these works did not
contain a complete enumerations.
A complete enumeration in the even orders up to~$44$ was given
by~\cite{DBLP:phd/basesearch/Bright17} and this was extended to order $64$ by~\cite{bright2017sat+}.

One reason why more attention has traditionally
been given to the odd order case is
due to the fact that if it was possible to construct Williamson
matrices in all odd orders this would resolve the
Hadamard conjecture.  On the other hand,
constructing Williamson matrices in all even orders
would not resolve the Hadamard conjecture
because Hadamard matrices constructed using
Williamson matrices of even order have orders which are divisible by~$8$.
However, it is still not even known if Hadamard
matrices exist for all orders divisible by~$8$, so nevertheless
studying Williamson
matrices of even order
has the potential to shed light on the Hadamard conjecture as well.

The Williamson conjecture was shown to be false by \cite{dokovic1993williamson}
who showed that such matrices do not exist in order~$35$.
Later, when an enumeration of Williamson matrices for odd orders $n<60$
was completed~\citep{holzmann2008williamson} it was found that Williamson matrices also do not
exist for orders $47$, $53$, and~$59$ but exist for all other odd orders under~$65$
since Turyn's construction~\citep{turyn1972infinite} works in orders $61$ and~$63$.

In this paper we provide for the first time
a complete enumeration of Williamson
matrices in the orders~$63$, $66$, $68$, $69$, and~$70$.
In particular, we show that Williamson matrices exist in every even order up to $70$.
This leads us to state what we call the \emph{even Williamson conjecture}:

\begin{conjecture}\label{conj:will}
Williamson matrices exist in every even order.
\end{conjecture}

The fact that Williamson matrices of even order turn out to be somewhat
plentiful gives some evidence for the truth of Conjecture~\ref{conj:will}.
Though we do not know how to prove Conjecture~\ref{conj:will}
our enumeration could potentially uncover structure in Williamson matrices
which might then be exploited in a proof of the conjecture.

Additionally, we point out that the existence of Williamson matrices
of order~$70=2\cdot35$ is especially interesting since $35$ is the
smallest order for which Williamson matrices do not exist.
Using complex Hadamard matrices, \cite{turyn1970} showed the existence
of Williamson matrices of odd order~$n$
implies the existence of Williamson matrices
of orders $2^kn$ for $k=1$, $2$, $3$, $4$.
Since Williamson matrices exist for all odd orders $n<35$
Turyn's result implies that Williamson matrices exist for all even orders
strictly less than $70$.
Since Williamson matrices of order $35$ do not exist
Turyn's result cannot be used to show the existence of Williamson
matrices of order $70$; the question of existence in
order $70$ was open until this paper.

We also determine that there are exactly two sets of Williamson matrices
(up to the equivalence given in Section~\ref{sec:willequiv}) of order $63$.
One of these falls under the aforementioned construction given
by~\cite{turyn1972infinite} while the other is new and is the first newly
discovered set of Williamson matrices in an odd order since one
was found using an exhaustive search in order $43$ by~\cite{holzmann2008williamson}.
In order $69$ our enumeration method produced just one set of Williamson
matrices and that set falls under the construction given by Turyn.


\subsection{Williamson matrices and sequences}\label{sec:background}

We now give the background on Williamson matrices
and their properties that are
necessary to understand the remainder of the paper.
The definition of Williamson matrices is motivated by the following theorem
used for constructing Hadamard matrices by~\cite{williamson1944hadamard}.

\begin{theorem}
\label{thm:williamson}
Let\/ $n \in \N$ and let\/ $A$, $B$, $C$, $D \in \{\pm 1\}^{n\times n}$.
Further, suppose that
\begin{enumerate}
\item $A$, $B$, $C$, and\/ $D$ are symmetric;
\item $A$, $B$, $C$, and\/ $D$ commute pairwise (i.e., $AB=BA$, $AC=CA$, etc.);
\item $A^2 + B^2 + C^2 + D^2 = 4nI_n$, where\/ $I_n$ is the identity
  matrix of order\/ $n$.
\end{enumerate}
Then 
\[
\begin{bmatrix}
A & B & C & D \\
-B & A & -D & C\\
-C & D & A & -B\\
-D & -C & B & A
\end{bmatrix}
\]
is a Hadamard matrix of order $4n$. 
\end{theorem}

To make the search for such matrices more tractable, and in particular
to make condition 2 trivial, Williamson also required the matrices
$A$, $B$, $C$, $D$ to be circulant matrices, as defined below.
\begin{definition}\label{def:circmatrix}
An\/ $n\times n$ matrix\/ $A=(a_{ij})$ is circulant if\/ $a_{ij}=a_{0,(j-i)\bmod n}$
for all\/ $i$ and\/ $j\in\brace{0,\dotsc,n-1}$.
\end{definition}
Circulant matrices $A$, $B$, $C$, $D$ which satisfy the conditions of Theorem~\ref{thm:williamson}
are known as a quadruple of \emph{Williamson matrices} in honour of Williamson.
Since Williamson matrices are circulant they
are defined in terms of their first row $[x_0,\dotsc,x_{n-1}]$
and since they are symmetric this row must be a symmetric sequence,
i.e., satisfy $x_i=x_{n-i}$ for $1\leq i<n$.
Given these facts, it is often convenient to work in terms of sequences rather than matrices.
When working with sequences in this context the
following function becomes very useful.
\begin{definition}
The \emph{periodic autocorrelation function} of the sequence\/ $A=[a_0,\dotsc,a_{n-1}]$ is
the function given by
\[ \PAF_A(s) \coloneqq \sum_{k=0}^{n-1}a_k a_{(k+s) \bmod n} . \] 
We also use\/ $\PAF_A$ to refer to a sequence
containing the values of the above function (which has period\/ $n$), i.e.,
\[ \PAF_A \coloneqq \brack[\big]{\PAF_A(0),\dotsc,\PAF_A(n-1)} . \]
\end{definition}
This function allows us to easily give a definition of Williamson matrices in terms of sequences.
\begin{definition}\label{def:williamsonsequences}
Four symmetric sequences\/ $A$, $B$, $C$, $D\in\brace{\pm1}^n$ are called a \emph{Williamson sequence quadruple}
(or simply \emph{Williamson sequences})
if they satisfy 
\begin{equation*}\label{eq:willdef}
\PAF_A(s) + \PAF_B(s) + \PAF_C(s) + \PAF_D(s) = 0
\end{equation*}
for\/ $s=1$, $\dots$, $\floor{n/2}$.
\end{definition}
It is straightforward to see that there is an equivalence between
such sequences and Williamson matrices~\citep[\S3.2]{DBLP:conf/casc/BrightGHKNC16}
because the first row of the matrix~$A^2$ is exactly the sequence~$\PAF_A$.
Therefore, for the remainder of this paper we will work directly with these sequences instead of
Williamson matrices.

\subsection{Williamson equivalences}\label{sec:willequiv}

Given a Williamson sequence quadruple $A$, $B$, $C$, $D$ of order $n$ there are five types of invertible operations which can be applied to produce another set of Williamson sequences, though two of the operations only apply when $n$ is even.  These operations allow us to define \emph{equivalence classes} of sets of Williamson sequences.  If a single Williamson sequence quadruple is known it is straightforward to generate all sets of Williamson sequences in the same equivalence class, so it suffices to search for Williamson sequences up to these equivalence operations.

\begin{enumerate}
\item[E1.] (Reorder) Reorder the sequences $A$, $B$, $C$, $D$ in any way.
\item[E2.] (Negate) Negate all the entries of any of $A$, $B$, $C$, or~$D$.
\item[E3.] (Shift) If $n$ is even, cyclically shift all the entries in any of $A$, $B$, $C$, or~$D$ by an offset of~$n/2$. 
\item[E4.] (Permute entries) Apply an automorphism of the cyclic group $C_n$ to all the indices of the entries of each of $A$, $B$, $C$, and~$D$ simultaneously.
\item[E5.] (Alternating negation) If $n$ is even, negate every second entry in each of $A$, $B$, $C$, and~$D$ simultaneously.
\end{enumerate}

These equivalence operations are well known~\citep{holzmann2008williamson}
except for the shift and alternating negation operations which have not traditionally
been used because they only apply when $n$ is even.
In fact, they were overlooked until a careful examination of the sequences produced by
our enumeration method.

\subsection{Fourier analysis}

We now give an alternative definition of Williamson sequences using concepts from Fourier analysis.
First, we define the power spectral density of a sequence.

\begin{definition}\label{def:psd}
The \emph{power spectral density} of the sequence\/ $A=[a_0,\dotsc,a_{n-1}]$ is
the function
\[ \PSD_A(s) \coloneqq \abs[\big]{\DFT_A(s)}^2 \]
where\/ $\DFT_A$ is the \emph{discrete Fourier transform}
of\/ $A$, i.e., $\DFT_A(s) \coloneqq \sum_{k=0}^{n-1} a_k e^{2\pi iks/n}$.
Equivalently, we may also consider the power spectral density
to be a sequence containing the values of the above function, i.e.,
\[ \PSD_A \coloneqq \brack[\big]{\PSD_A(0),\dotsc,\PSD_A(n-1)} . \]
\end{definition}

It now follows by~\cite[Theorem~2]{dokovic2015compression} that Williamson
sequences have the following alternative definition.

\begin{theorem}\label{thm:willdef}
Four symmetric sequences\/ $A$, $B$, $C$, $D\in\brace{\pm1}^n$ are Williamson sequences
if and only if
\begin{equation}\label{eq:willpsd}\tag{$*$}
\PSD_A(s) + \PSD_B(s) + \PSD_C(s) + \PSD_D(s) = 4n
\end{equation}
for\/ $s=0$, $\dots$, $\floor{n/2}$.
\end{theorem}
\begin{corollary}\label{cor:psdtest}
If\/ $\PSD_A(s)>4n$ for any value\/ $s$ then\/ $A$ cannot be part of a Williamson sequence.
\end{corollary}
\begin{proof}
Since $\PSD$ values are nonnegative, if $\PSD_A(s)>4n$ then the relationship~\eqref{eq:willpsd}
cannot hold and thus $A$ cannot be part of a Williamson sequence.
\end{proof}
Similarly, one can extend this so-called PSD test in Corollary~\ref{cor:psdtest} to apply to more than one sequence
at a time.
\begin{corollary}\label{cor:psdtestext}
If\/ $\PSD_A(s)+\PSD_B(s)>4n$ for any value of\/ $s$ then\/ $A$ and\/ $B$ do not occur
together in a Williamson sequence
and if\/ $\PSD_A(s)+\PSD_B(s)+\PSD_C(s)>4n$ for any value of\/ $s$
then\/ $A$, $B$, and\/ $C$ do not occur together in a Williamson sequence.
\end{corollary}
Additionally, a quadruple of sequences $A$, $B$, $C$, $D$ that are not Williamson
must necessarily fail the PSD test for some value of $s$.
\begin{corollary}\label{cor:necpsd}
If four sequences\/ $A$, $B$, $C$, $D\in\brace{\pm1}^n$ are not Williamson sequences
then\/
$\PSD_A(s) + \PSD_B(s) + \PSD_C(s) + \PSD_D(s) > 4n$
for some\/ $s$. 
\end{corollary}
\begin{proof}
By Parseval's theorem the average value of $\PSD_A(s) + \PSD_B(s) + \PSD_C(s) + \PSD_D(s)$
for $s=\{0,\dotsc,n-1\}$ is $4n$.  Thus either this sum is the constant $4n$ (in which
case the sequences are Williamson) or there is some~$s$ for which this sum is larger
than~$4n$.
\end{proof}

\subsection{Compression}\label{sec:compression}

As in the work by~\cite{dokovic2015compression} we now introduce the notion of \emph{compression}.
\begin{definition}
Let\/ $A = [a_0,a_1,\ldots,a_{n-1}]$ be a sequence of length\/ $n = dm$ and set 
\[ a_j^{(d)}=a_j + a_{j+d} + \dotsb + a_{j+(m-1)d}, \qquad j=0, \dotsc, d-1. 
\label{compression} \]
Then we say that the sequence\/ 
$A^{(d)} = [a_0^{(d)},a_1^{(d)},\ldots,a_{d-1}^{(d)}]$ 
is the \emph{$m$-compression} of\/ $A$.
\end{definition}

From~\cite[Theorem~3]{dokovic2015compression} we have the following result.

\begin{theorem}\label{thm:willcomp}
If\/ $A$, $B$, $C$, $D$ are Williamson sequences of order\/ $n$ then
\[ \PAF_{A'} + \PAF_{B'} + \PAF_{C'} + \PAF_{D'} = [4n,0,\dotsc,0] \]
and
\[ \PSD_{A'} + \PSD_{B'} + \PSD_{C'} + \PSD_{D'} = [4n,\dotsc,4n] \]
for any compression\/ $A'$, $B'$, $C'$, $D'$ of that set of Williamson sequences.
\end{theorem}

\begin{corollary}\label{cor:willrowsum}
If\/ $A$, $B$, $C$, $D$ are Williamson sequences of order\/ $n$ then
\begin{equation*}
R_A^2 + R_B^2 + R_C^2 + R_D^2 = 4n 
\label{eq:willdioeq}\tag{$**$}
\end{equation*}
where\/ $R_X$ denotes the rowsum of\/ $X$.
\end{corollary}
\begin{proof}
Let $X'$ be the $n$-compression of $X\in\brace{\pm1}^n$, i.e., $X'$ is a sequence with one
entry whose value is $R_X$.  Note that $\PSD_{X'}=[R_X^2]$, so the result follows
by Theorem~\ref{thm:willcomp}. 
\end{proof}

\subsection{Williamson's product theorem}

\cite{williamson1944hadamard} proved the following theorem:

\begin{theorem}\label{thm:willprododd}
If\/ $A$, $B$, $C$, $D$ are Williamson sequences of odd order\/ $n$ then
\[a_ib_ic_id_i=-a_0b_0c_0d_0 \qquad \text{for\/ $1\leq i<n/2$.} \]
\end{theorem}

We prove a version of this theorem for even $n$:

\begin{theorem}\label{thm:willprodeven}
If\/ $A$, $B$, $C$, $D$ are Williamson sequences of even order\/ $n=2m$ then
\[a_ib_ic_id_i=a_{i+m}b_{i+m}c_{i+m}d_{i+m} \qquad \text{for\/ $0\leq i<m$.} \]
\end{theorem}

Although this theorem is not an essential part of our algorithm it improves its
efficiency by allowing us to cut down the size of the search space.
Our algorithm uses the theorem in the following form:

\begin{corollary}\label{cor:willprod}
If\/ $A'$, $B'$, $C'$, $D'$ are the\/ $2$\nobreakdash-compressions of a set of Williamson sequences
then\/ $A'+B'+C'+D'\equiv[0,\dotsc,0]\pmod{4}$.
\end{corollary}

Proofs of Theorem~\ref{thm:willprodeven} and Corollary~\ref{cor:willprod}
are available in the appendix.

\subsection{Doubling construction}\label{sec:dblconstr}

We now give a simple construction that generates Williamson sequences of
order~$2n$ from Williamson sequences of odd order $n$ using the following
three operations on sequences $A=[a_0,\dotsc,a_{n-1}]$ and $B=[b_0,\dotsc,b_{n-1}]$:
\begin{enumerate}
\item Negation.  Individually negate each entry of $A$, i.e., $-A\coloneqq[-a_0,\dotsc,-a_{n-1}]$.
\item Shift.  Cyclically shift the entries of $A$ by an offset of $(n-1)/2$,
i.e., 
\[ A' \coloneqq [a_{(n-1)/2},\dotsc,a_{n-1},a_0,a_1,\dotsc,a_{(n-3)/2}] . \]
\item Interleave.  Interleave the entries of $A$ and $B$ in a perfect shuffle, i.e.,
\[ A\shuffle B \coloneqq [a_0,b_0,a_1,b_1,\dotsc,a_{n-1},b_{n-1}] . \]
\end{enumerate}
Our doubling construction is captured by the following theorem.
\begin{theorem}\label{thm:willdbl}
Let $A$, $B$, $C$, $D$ be Williamson sequences of odd order $n$.
Then
{\rm\[ A\shuffle B',\, (-A)\shuffle B',\, C\shuffle D',\, (-C)\shuffle D' \]}%
are Williamson sequences of order $2n$.
\end{theorem}

We remark that a single set of Williamson sequences of order $n$
can often be used to generate more
than one set of Williamson sequences of order~$2n$ by applying equivalence operations
to the quadruple $A$, $B$, $C$, $D$ before using the construction.
For example, the single inequivalent set of Williamson sequences of order 5
can be used to generate both inequivalent sets of Williamson sequences of order 10
using this construction with an appropriate reordering of $A$, $B$, $C$, $D$. 

We also remark that this doubling construction can be reversed in the sense that
if Williamson sequences of order $2n$ exist for $n$ odd then symmetric sequences
$X_1$, $\dotsc$, $X_8\in\brace{\pm1}^n$ 
can be constructed that satisfy the Williamson
property
\[ \sum_{i=1}^8 \PAF_{X_i}(s) = 0 \qquad\text{for $s=1$, $\dotsc$, $n-1$} . \]
We call such sequences \emph{8-Williamson sequences} because they form the
first rows of \emph{8-Williamson matrices} as defined by for example
\cite{kotsireas2006constructions}.  Note that the equivalence operations
of Section~\ref{sec:willequiv} also define an equivalence on 8-Williamson
sequences so long as they are written to apply to~8 sequences instead of~4.

\begin{theorem}\label{thm:8will}
Let $A$, $B$, $C$, $D$ be Williamson sequences of order $2n$
with $n$ odd and write
{\rm\[ A = A_1\shuffle A_2',\, B = B_1\shuffle B_2',\, C = C_1\shuffle C_2',\, D = D_1\shuffle D_2' . \]}%
Then $A_1$, $A_2$, $B_1$, $B_2$, $C_1$, $C_2$, $D_1$, $D_2$ are 8-Williamson sequences of order $n$.
\end{theorem}


Proofs of Theorems~\ref{thm:willdbl} and~\ref{thm:8will} are available
in the appendix.  Similar constructions have been described for
complementary sets of sequences by \cite{tseng1972complementary}
but to our knowledge these constructions are new in the context
of Williamson sequences.

\section{The SAT+CAS paradigm}\label{sec:sat+cas}

The idea of combining SAT solvers with computer algebra systems
originated independently in two works published in 2015:
In a paper at the conference CADE entitled ``\textsc{MathCheck}: A Math Assistant
via a Combination of Computer Algebra Systems and SAT Solvers'' by~\cite{DBLP:conf/cade/ZulkoskiGC15}
and in an invited talk at the conference ISSAC entitled
``Building Bridges between Symbolic Computation and Satisfiability Checking''
by~\cite{DBLP:conf/issac/Abraham15}.
The paradigm was also anticipated by \cite{jovanovic2012solving}
who used CAS techniques in SAT-like 
search algorithms. 

The CADE paper describes a tool called \MC\ that combines
the general-purpose search capability of SAT solvers with the domain-specific
knowledge of computer algebra systems.
The paper made the case that \MC\
\begin{quote}
\dots combines the efficient search routines of modern
SAT solvers, with the expressive power of CAS, thus complementing
both.
\end{quote}
As evidence for the power of this paradigm, they used
\MC\ to improve the best known bounds in two conjectures in graph theory.

Independently, the computer scientist Erika \'Abrah\'am observed
that the fields of satisfiability checking and symbolic computation share
many common aims but in practice are quite separated, with little communication
between the fields:
\begin{quote}
\ldots collaboration between symbolic computation and SMT [SAT modulo theories]
solving is still (surprisingly) quite restricted\ldots
\end{quote}
Furthermore, she outlined reasons why combining the insights from both fields
had the potential to solve certain problems more efficiently than would be
otherwise possible.  To this end, the SC$^2$ project~\citep{sc2} was started
with the aim of fostering collaboration between the two communities.

\subsection{Programmatic SAT and SMT}

Not all constraints can easily be expressed in SAT instances.
To deal with this, sophisticated \emph{SMT solvers} were developed that can determine
the satisfiability of formulas in first-order logic
with respect to certain logical theories~\citep{barrett2009satisfiability}.
SMT solvers are often based on the
Davis--Putnam--Logemann--Loveland algorithm (modulo theories) by \cite{ganzinger2004dpll}
denoted DPLL($T$) where~$T$ is a theory of first-order logic.
In this framework SAT solvers handle the Boolean aspects of a formula
while~$T$ solvers handle the theory aspects of a formula.
\cite{DBLP:conf/issac/Abraham15} gives a number of requirements that a theory
solver should satisfy, and the SMT-LIB standard by \cite{BarFT-SMTLIB}
describes many different theories used by SMT solvers along with a common
input and output language for those theories.

Programmatic SAT solvers were introduced by \cite{DBLP:conf/sat/GaneshOSDRS12}
and are simply a variant of DPLL($T$)
where the~$T$ solvers are replaced by arbitrary user-specified code.
However, programmatic SAT solvers only deal with Boolean logic and
not first-order logic, similar to the
``SAT modulo SAT'' solvers by \cite{bayless2013efficient}.
Additionally, the code provided to a programmatic SAT solver
can be specialized to individual formulas, while the $T$ solvers in
DPLL($T$) deal with a specific theory of first-order logic
that is typically fixed in advance.  For our purposes in this paper we will use a CAS
as a theory solver and therefore
our usage of programmatic SAT can be considered an instance of DPLL(CAS).

A programmatic SAT solver can generate conflict clauses programmatically, i.e.,
by a piece of code that runs as the SAT solver carries out its search.
Such a SAT solver can learn clauses that are more useful than the conflict
clauses that it learns by default.  Not only can this make the SAT solver's search more efficient,
it allows for increased expressiveness as many types of constraints that are 
awkward to express in a conjunctive normal form format can naturally be expressed
using code.  Additionally, it allows one to compile \emph{instance-specific} SAT solvers
which are tailored to solving one specific type of instance.  In this framework
instances no longer have to solely consist of a set of clauses in conjunctive normal form.
Instead, instances can consist of both a set of CNF clauses and a piece of code that
encodes constraints that are too cumbersome to be written in CNF format.

As an example of this, consider the case of searching for Williamson sequences using 
a SAT solver.  One could encode Definition~\ref{def:williamsonsequences} in CNF
format by using Boolean variables to represent the entries in the Williamson sequences and by
using binary adders to encode the summations; such a method was used by~\cite{DBLP:conf/casc/BrightGHKNC16}.
However, one could also use the equivalent definition given in Theorem~\ref{thm:willdef}.
This alternate definition has the advantage that it becomes easy to apply
Corollaries~\ref{cor:psdtest} and~\ref{cor:psdtestext}, which allows one to filter many
sequences from consideration and greatly speed up the search.
Because of this, our method will use the constraints~\eqref{eq:willpsd}
from Theorem~\ref{thm:willdef}
to encode the definition of Williamson sequences in our SAT instances.

However, encoding the equations in~\eqref{eq:willpsd} would be extremely cumbersome
to do using CNF clauses, because of the involved nature of computing the PSD values.
However, the equations~\eqref{eq:willpsd} are easy to express programmatically---%
as long as one has a method of computing the PSD values.  This can be done
efficiently using the fast Fourier transform which is available in many
computer algebra systems and mathematical libraries.

Thus, our SAT instances will not use CNF clauses to encode the defining property of
Williamson sequences but instead encode those clauses programmatically.
This is done by writing a \emph{callback function}
that is compiled with the SAT solver and programmatically expresses the constraints
in Theorem~\ref{thm:willdef} and
the filtering criteria of Corollaries~\ref{cor:psdtest}
and~\ref{cor:psdtestext}.

\subsection{Programmatic Williamson encoding}\label{sec:progwill}

We now describe in detail our programmatic encoding of Williamson sequences.
The encoding takes the form of a piece of code which examines a partial assignment
to the Boolean variables defining the sequences $A$, $B$, $C$, and~$D$ (where true
encodes~$1$ and false encodes~$-1$).
In the case when the partial assignment can be ruled out using
Corollaries~\ref{cor:psdtest} or~\ref{cor:psdtestext}, a conflict
clause is returned which encodes a reason why the partial assignment
no longer needs to be considered.  If the sequences actually form a Williamson
sequence then they are recorded in an auxiliary file; at this point the
solver would traditionally return SAT and stop, though our implementation continues the search
because we 
want to do a
complete enumeration of the space.

The programmatic callback function does the following:

\begin{enumerate}
\item Initialize $S\coloneqq\emptyset$.  This variable will be a set which contains
the sequences whose entries are all currently assigned.
\item Check if all the variables encoding the entries in sequence~$A$ have been assigned;
if so, add $A$ to the set~$S$ and compute $\PSD_A$, otherwise skip to the next step.
When $\PSD_A(s)>4n$
for some value of $s$ then learn a clause prohibiting the entries of $A$ from being
assigned the way they currently are, i.e., learn the clause
\begin{equation*}
\lnot(a_0^{\text{cur}}\land a_1^{\text{cur}} \land\dotsb\land a_{n-1}^{\text{cur}}) \equiv \lnot a_0^{\text{cur}}\lor\lnot a_1^{\text{cur}}\lor\dotsb\lor\lnot a_{n-1}^{\text{cur}}
\end{equation*}
where $a_i^{\text{cur}}$ is the literal $a_i$ when $a_i$ is currently assigned to true
and is the literal $\lnot a_i$ when $a_i$ is currently assigned to false.
\item Check if all the variables encoding the entries in sequence~$B$ have been assigned;
if so, add $B$ to the set~$S$ and compute $\PSD_B$.  When there is some $s$ such that $\sum_{X\in S}\PSD_X(s)>4n$ then learn a clause
prohibiting the values of the sequences in~$S$ from being assigned the way they currently are.
\item Repeat the last step again twice, once with $B$ replaced with $C$ and then again with~$B$ replaced with $D$. 
\item If all the variables in sequences $A$, $B$, $C$, and $D$ are assigned then
record the sequences in an auxiliary file
and learn a clause prohibiting the values of the sequences from being
assigned the way they currently are so that this assignment is not examined again.
\end{enumerate}

After the search is completed the auxiliary file will contain
all sequences that passed the PSD tests and thus all Williamson sequences
will be in this list.  Verifying a sequence is in fact Williamson can be done using
Definition~\ref{def:williamsonsequences}.  (If the PSDs were computed exactly
this is not necessary by Corollary~\ref{cor:necpsd} but with
floating-point arithmetic there is a slight chance that
a PSD test failure was not detected.)
Note that the clauses learned by this function allow the SAT solver to 
execute the search significantly faster than would be
possible using a brute-force technique.  As a rough estimate
of the benefit, note that there are approximately
$2^{n/2}$ possibilities for each member $A$, $B$, $C$, $D$ in a set of Williamson
sequences.  If no clauses are learned in steps~2--4 then the SAT solver
will examine all $2^{4(n/2)}$ total possibilities.
Conversely, if a clause is always learned in step~2 then the SAT solver
will only need to examine the $2^{n/2}$ possibilities for~$A$.
Of course, one will not always learn a clause in steps~2\nobreakdash--4
but in practice such a clause is learned quite frequently
and this more than makes up for the overhead of
computing the PSD values (this accounted for about 20\%
of the SAT solver's runtime in our experiments).
The programmatic approach
was essential for the largest orders that we were able to solve;
see Table~\ref{tbl:satresults} in Section~\ref{sec:results}
for a comparison between the running times
of a SAT solver using the CNF and programmatic encodings.
However, the programmatic encoding was much too slow to be able perform the enumeration
by itself and a more sophisticated enumeration algorithm
was required (using the programmatic encoding as a subroutine).

\section{Our enumeration algorithm}\label{sec:sat+casmethod}

A high-level summary of the components of our enumeration algorithm
is given in Figure~\ref{fig:diagram}.  We require two kinds of
functions from computer algebra systems or mathematical libraries,
namely, one that can solve the quadratic Diophantine equation~\eqref{eq:willdioeq}
and one that can compute the discrete Fourier transform of a
sequence. 

In the following description we have step~1 handled by the
Diophantine equation solver, steps~2--4 handled by
the driver script, and step~5 handled by the programmatic
SAT solver.  The driver script is responsible for constructing
the SAT instances and passing them off to the programmatic SAT solver.
It also implicitly passes encoding information to the system responsible
for performing the programmatic Williamson encoding described in
Section~\ref{sec:progwill}, i.e., the system needs to know
which variables in the SAT instance correspond to which
Williamson sequence entries, but this is fixed in advance.

We now give a complete description of our method that enumerates
all Williamson sequences of a given order~$n$ divisible by~$2$ or~$3$.
Let $m\in\brace{2,3}$ be the smallest prime divisor of $n$.

\begin{figure}\centering
\begin{tikzpicture}
\tikzset{>=latex'}
\tikzset{auto}
\tikzstyle{block} = [draw, rectangle, minimum height=3em, minimum width=6.975em, text width=6.975em, align=center]
\node (input) {\small $n\mspace{1mu}$};
\node [block,right of=input,node distance=2cm] (gen) {\small Driver script};
\node [block,right of=gen,node distance=5.5cm] (sat) {\small Programmatic\\SAT solver};
\node [block,above of=gen,node distance=2.75cm] (cas) {\small Diophantine solver\\Fourier transform};
\node [block,above of=sat,node distance=2.75cm] (cas2) {\small Fourier transform};
\draw [draw,->] (input) -- (gen);
\draw [draw,->] (sat.140) -- node[text width=1.375cm,align=center] {\small Partial\\assignment\\} (cas2.220);
\draw [draw,<-] (sat.40) -- node[right,text width=1cm,align=center] {\small Conflict\\clause\\} (cas2.320);
\draw [draw,->] (gen.140) -- node[text width=1.125cm,align=center] {\small External\\call\\} (cas.220);
\draw [draw,<-] (gen.40) -- node[right,text width=0.75cm,align=center] {\small Result} (cas.320);
\draw [draw,->] (gen) -- node[below] {\small SAT instances} (sat);
\draw [draw,->,dashed] (gen.12.5) -| (4.25,2.75) --node[text width=4cm,align=center]{\small Encoding\\information\\} (cas2);
\node [right of=sat,node distance=2.75cm,text width=4.5em,align=center] (output) {\small Enumeration\\in order $n$};
\draw [draw,->] (sat) -- (output);
\end{tikzpicture}
\caption{Outline of our algorithm for enumerating Williamson sequences of order $n$.
The boxes on the left correspond to the preprocessing which encodes and decomposes
the original problem into SAT instances.  The boxes
on the right correspond to an SMT-like setup where the
system that computes the discrete Fourier transform
takes on the role of the theory solver.}\label{fig:diagram}\end{figure}
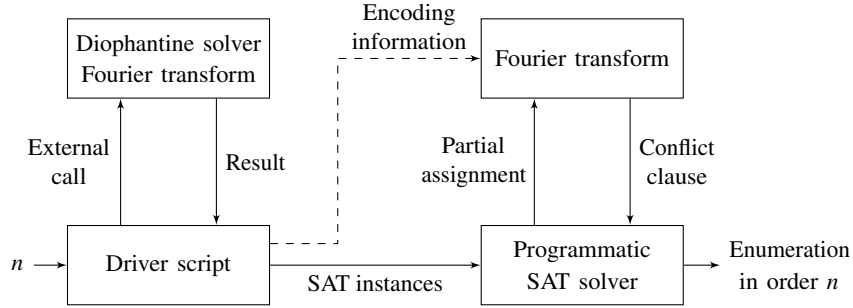

\subsection{Step 1: Generate possible sum-of-squares decompositions}

First, note that by Corollary~\ref{cor:willrowsum} every set of Williamson sequences
gives rise to a decomposition of $4n$ into a sum of four squares.
We query a computer algebra system such as \textsc{Maple}
or \textsc{Mathematica} to get all possible solutions of 
the Diophantine equation~\eqref{eq:willdioeq}.
Because we only care about Williamson sequences up to equivalence,
we add the inequalities
\[ 0 \leq R_A \leq R_B \leq R_C \leq R_D \]
to the Diophantine equation; it is clear that any Williamson
sequence quadruple can be transformed into another quadruple
that satisfies these inequalities by applying the reorder and/or
negating equivalence operations.

\subsection{Step 2: Generate possible Williamson sequence quadruple members}

Next, we form a list of the sequences that could possibly appear as a member of
a Williamson sequence quadruple of order~$n$.
To do this, we examine every symmetric sequence $X\in\brace{\pm1}^n$.
For all such~$X$ we compute $\PSD_X$ and ignore those that satisfy
$\PSD_X(s)>4n$ for some $s$.
We also ignore those~$X$ whose rowsum does not appear in any possible
solution $(R_A,R_B,R_C,R_D)$ of the sum-of-squares
Diophantine equation~\eqref{eq:willdioeq}.
The sequences~$X$ that remain after this process form a list of the sequences
that could possibly appear as a member of a set of Williamson sequences.
At this stage we could generate all Williamson sequences quadruples of order~$n$
by trying all ways of grouping the possible sequences~$X$ into quadruples and
filtering those that are not Williamson.  However, because of the large number
of ways in which this grouping into quadruples can be done
this is not feasible to do except for small~$n$.

\subsection{Step 3: Perform compression}

In order to reduce the size of the problem so that the possible
sequences generated in step~2 can be grouped into quadruples we first
compress the sequences using the process
described in Section~\ref{sec:compression}.
For each solution $(R_A,R_B,R_C,R_D)$ of the sum-of-squares Diophantine equation~\eqref{eq:willdioeq}
we form four lists $L_A$, $L_B$, $L_C$, and $L_D$.  The list $L_A$ will contain
the $m$\nobreakdash-compressions of the sequences $X$ generated in step~2 that have rowsum~$R_A$
(and the other lists
will be defined in a similar manner).
Note that the sequences in these lists will be $\brace{\pm2,0}$-sequences
if $n$ is even and $\brace{\pm3,\pm1}$-sequences if $n$ is odd
since they are $m$\nobreakdash-compressions of the sequences $X$
which are $\brace{\pm1}$-sequences.

\subsection{Step 4: Match the compressions}

By construction, the lists $L_A$, $L_B$, $L_C$, and $L_D$ contain all possible
$m$\nobreakdash-compressions of the members of Williamson sequence quadruples whose sum-of-squares
decomposition is $R_A^2+R_B^2+R_C^2+R_D^2$.  Thus, by trying all possible sum-of-squares
decompositions and all ways of matching
together the sequences from the lists $L_A$, $L_B$, $L_C$, $L_D$ we can find all
$m$\nobreakdash-compressions of Williamson sequence quadruples of order~$n$.
By Theorem~\ref{thm:willcomp}, a necessary condition for $A$, $B$, $C$, $D$
to be Williamson sequences is that
\[ \PSD_{A'}+\PSD_{B'}+\PSD_{C'}+\PSD_{D'} = [4n,\dotsc,4n] \]
where $A'$, $B'$, $C'$, $D'$ are the $m$\nobreakdash-compressions of $A$, $B$, $C$, $D$.
Therefore, one could perform this step by enumerating all
$(A',B',C',D')\in L_A\times L_B\times L_C\times L_D$
and outputting those whose PSDs sum to $[4n,\dotsc,4n]$ as a potential $m$\nobreakdash-compression
of a Williamson sequence quadruple.
However, there will typically be far too many elements of
$L_A\times L_B\times L_C\times L_D$ to try in a reasonable amount
of time.

Instead, we will enumerate all $(A',B')\in L_A\times L_B$ and $(C',D')\in L_C\times L_D$
and use a string sorting technique by~\cite{DBLP:journals/ol/KotsireasKP10} to find which
$(A',B')$ and $(C',D')$ can be matched together to form potential
$m$\nobreakdash-compressions of Williamson sequences.
To determine which pairs can be matched together we use the necessary
condition from Theorem~\ref{thm:willcomp} in a slightly rewritten form,
\[ \PAF_{A'}+\PAF_{B'} = [4n,0,\dotsc,0]-(\PAF_{C'}+\PAF_{D'}) . \]
Our matching procedure outputs a list of the $(A',B',C',D')$
which satisfy this condition, and therefore output a list of potential
$m$\nobreakdash-compressions of Williamson sequences.

In detail, our matching procedure performs the following steps:
\begin{algorithmic}[1]
\State \textbf{initialize} $L_{AB}$ and $L_{CD}$ to empty lists
	\For{$(A',B')\in L_A\times L_B$}
		\If{$\PSD_{A'}(s)+\PSD_{B'}(s)<4n$ for all $s$}
			\State \textbf{add} $\PAF_{A'}+\PAF_{B'}$ to $L_{AB}$
		\EndIf
	\EndFor
	\For{$(C',D')\in L_C\times L_D$}
		\If{$\PSD_{C'}(s)+\PSD_{D'}(s)<4n$ for all $s$}
			\State \textbf{add} $[4n,0,\dotsc,0]-(\PAF_{C'}+\PAF_{D'})$ to $L_{CD}$
		\EndIf
	\EndFor
\For{each $X$ common to both $L_{AB}$ and $L_{CD}$}
	\State \textbf{output} $(A',B')$ and $(C',D')$ which $X$ was generated from in an auxiliary file
\EndFor
\end{algorithmic}

Line~8 can be done efficiently by sorting the lists $L_{AB}$ and $L_{CD}$
and then performing a linear scan through the sorted lists to find the elements
common to both lists.  Line~9 can be done efficiently if with each element in
the lists $L_{AB}$ and $L_{CD}$ we also keep track of a pointer to the sequences
$(A',B')$ or $(C',D')$ that the element was generated from in line~4 or~7.
Also in line~9 if $n$ is even we only
output sequences for which $A'+B'+C'+D'$ is the zero vector mod~$4$ as this
is an invariant of all $2$\nobreakdash-compressed Williamson sequences by Corollary~\ref{cor:willprod}.

\subsection{Step 5: Uncompress the matched compressions}

It is now necessary to find the Williamson sequences, if any, which when compressed
by a factor of~$m$ produce one of the sequence quadruples generated in step~4.
In other words, we want to find a way to perform uncompression on the matched compressions
which we generated.  To do this, we formulate the uncompression problem as a Boolean
satisfiability instance and use a SAT solver's combinatorial search facilities
to search for solutions to the uncompression problem.

We use Boolean variables to represent the entries of the uncompressed
Williamson sequences, with true representing the value of $1$ and false representing
the value of~$-1$.  Since Williamson sequence quadruples consist of four sequences of length~$n$
they contain a total of $4n$ entries, namely,
\[ a_0,\dotsc,a_{n-1},b_0,\dotsc,b_{n-1},c_0,\dotsc,c_{n-1},d_0,\dotsc,d_{n-1} . \]
However, because Williamson sequences are symmetric we actually only need
to define the distinct variables
\[ a_0,\dotsc,a_{\floor{n/2}},b_0,\dotsc,b_{\floor{n/2}},c_0,\dotsc,c_{\floor{n/2}},d_0,\dotsc,d_{\floor{n/2}} . \]
Any variable $x_i$ with $i>n/2$ can simply be replaced with the equivalent variable
$x_{n-i}$; in what follows we implicitly use this substitution when necessary.
Thus, the SAT instances that we generate will contain
$2n+4$ variables when $n$ is even and $2n+2$ variables when $n$ is odd.

Say that $(A',B',C',D')$ is one of the 
$m$\nobreakdash-compressions generated in step~4.
By the definition of $m$\nobreakdash-compression, we have that
$a'_i = a_i + a_{i+n/2}$ if $n$ is even and $a'_i= a_i + a_{i+n/3} + a_{i+2n/3}$
if $n$ is odd.
Since $-3\leq a'_i\leq 3$ there are
seven possibilities we must consider for each $a'_i$.

Case 1. If $a'_i=3$ then we must have $a_i=1$, $a_{i+n/3}=1$, and $a_{i+2n/3}=1$.
Thinking of the entries as Boolean variables, we add the clauses
\[ a_i \land a_{i+n/3} \land a_{i+2n/3} \]
to our SAT instance.

Case 2. If $a'_i=2$ then we must have $a_i=1$ and $a_{i+n/2}=1$.  Thinking of the
entries as Boolean variables, we add the clauses
\[ a_i \land a_{i+n/2} \]
to our SAT instance.

Case 3. If $a'_i=1$ then exactly one of the entries
$a_i$, $a_{i+n/3}$, and~$a_{i+2n/3}$ must be $-1$.
Thinking of the entries as Boolean variables, we add the clauses
\[ (\lnot a_i\lor\lnot a_{i+n/3}\lor\lnot a_{i+2n/3})\land(a_i\lor a_{i+n/3})\land(a_i\lor a_{i+2n/3})\land(a_{i+n/3}\lor a_{i+2n/3}) \]
to our SAT instance.  These clauses specify in conjunctive normal form
that exactly one of the variables $a_i$, $a_{i+n/3}$, and~$a_{i+2n/3}$ is false.

Case 4. If $a'_i=0$ then we must have $a_i=1$ and $a_{i+n/2}=-1$ or vice versa.
Thinking of the entries as Boolean variables, we add the clauses
\[ (a_i \lor a_{i+n/2}) \land (\lnot a_i \lor \lnot a_{i+n/2}) \]
to our SAT instance.  These clauses specify in conjunctive normal form
that exactly one of the variables $a_i$ and $a_{i+n/2}$ is true.

Case 5. If $a'_i=-1$ then exactly one of the entries
$a_i$, $a_{i+n/3}$, and~$a_{i+2n/3}$ must be $1$.
Thinking of the entries as Boolean variables, we add the clauses
\[ (a_i\lor a_{i+n/3}\lor a_{i+2n/3})\land(\lnot a_i\lor\lnot a_{i+n/3})\land(\lnot a_i\lor\lnot a_{i+2n/3})\land(\lnot a_{i+n/3}\lor\lnot a_{i+2n/3}) \]
to our SAT instance.  These clauses specify in conjunctive normal form
that exactly one of the variables $a_i$, $a_{i+n/3}$, and~$a_{i+2n/3}$ is true.

Case 6. If $a'_i=-2$ then we must have $a_i=-1$ and $a_{i+n/2}=-1$.  Thinking of the
entries as Boolean variables, we add the clauses
\[ \lnot a_i \land \lnot a_{i+n/2} \]
to our SAT instance.

Case 7. If $a'_i=-3$ then we must have $a_i=-1$, $a_{i+n/3}=-1$, and $a_{i+2n/3}=-1$.
Thinking of the entries as Boolean variables, we add the clauses
\[ \lnot a_i \land \lnot a_{i+n/3} \land \lnot a_{i+2n/3} \]
to our SAT instance.

For each entry $a'_i$ in $A'$ we add the clauses from the
appropriate case to the SAT instance, as well
as add clauses from a similar case analysis for the entries
from $B'$, $C'$, and $D'$.
A satisfying assignment to the generated SAT instance provides an uncompression
$(A,B,C,D)$ of $(A',B',C',D')$.  However, the uncompression need not be a set of
Williamson sequences.  To ensure that the solutions produced by the SAT solver
are in fact Williamson sequences we additionally use the programmatic SAT
Williamson encoding as described in Section~\ref{sec:progwill}.

For each $(A',B',C',D')$ generated in step~4 we generate a SAT instance which
contains the clauses specified above.  We then solve the SAT instances with
a programmatic SAT solver whose programmatic clause generator specifies
that any satisfying assignment of the instance encodes a set of Williamson sequences
and performs an exhaustive search to find all solutions.
By construction, every Williamson sequence quadruple of order~$n$ will have its
$m$\nobreakdash-compression generated in step~4, making this search
totally exhaustive (up to the discarded equivalences).

\subsection{Postprocessing: Remove equivalent Williamson sequences}\label{sec:equivcheck}

After step~5 we have produced a list of all the Williamson sequences of order $n$
that have a certain sum-of-squares decompositions.  We chose the
decompositions in such a way that every Williamson sequence quadruple will be equivalent to
at least one decomposition but we have not dealt with all equivalences E1--E5. Therefore
some Williamson sequences that we generate may be equivalent to each other.
Removing equivalences can be performed in step~5 but since the SAT instances
are generally solved in parallel we wait until all SAT instances have been solved
for simplicity.

For the purpose of counting the total number of inequivalent
Williamson sequences that exist in order $n$ it is necessary to
examine each quadruple in the list generated in step~5
and determine if it is equivalent to another quadruple in the list.  This can be done
by repeatedly applying the equivalence operations from Section~\ref{sec:willequiv}
on the quadruples in the list
and discarding those which are equivalent to a previously
found set of Williamson sequences.
However, this can be inefficient because there are
typically a large number of Williamson sequence quadruples in each
equivalence class.  Instead, a more efficient way of testing for equivalence
is to define a single representative in each equivalence class
that is easy to compute.  Then two sets of Williamson sequences can
be tested for equivalence by testing that their representatives are equal.

As a first step in defining a single representative
in each equivalence class of Williamson
sequence quadruples we first consider only the equivalence operations E1, E2, and~E3 (reorder,
negate, and shift).  The operations E2 and~E3 apply to individual sequences $X$ and
there are up to four sequences which could be generated using E2 and~E3, namely,
$X$, $\E2(X)$, $\E3(X)$, and $\E2(\E3(X))$.  Let $M_X$ be the lexicographic
minimum of these four sequences.  Given a Williamson sequence quadruple $(A,B,C,D)$, we compute
$(M_A,M_B,M_C,M_D)$ and then use operation E1 on the sequences in the quadruple
to sort those sequences in increasing lexicographic order.  The resulting quadruple
is the lexicographic minimum of all quadruples equivalent to $(A,B,C,D)$
using the operations E1, E2, and~E3 and is therefore a unique single representative of the
equivalence class which we denote $M_{(A,B,C,D)}$.

Next, consider the equivalence operation~E4 (permute entries).  Let~$\sigma$
be an automorphism of the cyclic group~$C_n$ and let $\sigma(X)$ be the sequence
whose $i$th entry is $x_{\sigma(i)}$.  Then the lexicographic minimum of the set
\[ S_{(A,B,C,D)} \coloneqq \brace[\big]{\, M_{(\sigma(A),\sigma(B),\sigma(C),\sigma(D))} : \sigma\in\Aut(C_n) \,} \]
is the lexicographic minimum of all quadruples equivalent to $(A,B,C,D)$
using the operations E1, E2, E3, and~E4.  (This is due to the fact that
E4 commutes with E1, E2, and~E3, so it is always possible to find the
global lexicographic minimum by first trying all possible ways of applying~E4
and only afterwards considering E1, E2, and~E3.)

Finally, if $n$ is even we consider the equivalence operation~E5 (alternating negation).
The lexicographic minimum of the set
$S_{(A,B,C,D)} \cup S_{\E5(A,B,C,D)}$
will be the lexicographic minimum of all quadruples equivalent to $(A,B,C,D)$
and is therefore a single unique representative of the equivalence class.
(Again, this is due to the
fact that E5 commutes with the other operations so it is always possible
to find the global lexicographic minimum by
first trying all possible ways of applying~E5 before applying the other operations.)

\subsection{Optimizations}\label{sec:optimizations}

While the procedure just described will correctly enumerate all Williamson sequences of
a given even order~$n$, there are a few optimizations that can be used to improve
the efficiency of the search.
Note that in step~3 we have not generated \emph{all} possible
$m$\nobreakdash-compression quadruples; we only generate those quadruples that have rowsums
$(R_A,R_B,R_C,R_D)$ that correspond to solutions of~\eqref{eq:willdioeq}, and we
use the negation and reordering equivalence operations to cut down the number
of possible rowsums necessary to check.  However, there still remain equivalences that
can be removed; if~$\sigma$ is an automorphism of the cyclic group $C_n$ then
$(A,B,C,D)$ is a Williamson sequence quadruple if and only if $(\sigma(A),\sigma(B),\sigma(C),\sigma(D))$
is a Williamson sequence quadruple (with $\sigma(X)$ defined so that its $i$th entry is $x_{\sigma(i)}$).
Thus if both $A$ and $\sigma(A)$ are in the list generated
in step~2 we can remove one from consideration.  Unfortunately, we cannot do the same
in the lists for $B$, $C$, and~$D$, since it is not possible to know which
representatives for $B$, $C$, and~$D$ to keep, as the representatives 
must match with the representative for $A$ that was kept.  

Similarly, in step~5 one can ignore any SAT instance that can be transformed
into another SAT instance using the equivalence operations from Section~\ref{sec:willequiv}.
In this case the solutions in the ignored SAT instance will be equivalent to those
in the SAT instance associated to it through the equivalence transformation.

In the programmatic Williamson encoding we can often learn shorter clauses
with a slight modification of the procedure described in Section~\ref{sec:progwill}.
Instead of checking $\sum_{X\in S}\PSD_X(s)>4n$ directly we instead
find the smallest subset $S'$ of $S$ such that $\sum_{X\in S'}\PSD_X(s)>4n$
(if such a subset exists).  This is done by sorting the values
of $\PSD_X(s)$ 
and performing the check using the largest values
$\PSD_X(s)$ before considering the smaller values.  For example,
if $\PSD_B(s)>\PSD_A(s)$ then we would check $\PSD_B(s)>4n$ before checking
$\PSD_A(s)+\PSD_B(s)>4n$.

When $n$ is odd we use Theorem~\ref{thm:willprododd} to provide additional
information to the SAT solver.  For simplicity, suppose we fix $a_0=1$;
as shown in~\citep[\S3.1.2]{DBLP:phd/basesearch/Bright17}
this can be done by fixing the sign of $\rowsum(A)$ to not necessarily
be positive but to satisfy $\rowsum(A)\equiv n\pmod{4}$.
Briefly, this is because the symmetry of $A$ implies that the rowsum of $A$
will always be congruent to $a_0+n-1\pmod{4}$.
Also fixing the values $b_0=c_0=d_0=1$, Theorem~\ref{thm:willprododd} says that
\[ a_kb_kc_kd_k = -1 \qquad \text{for $k=1$, $\dotsc$, $(n-1)/2$} . \]
Thinking of the entries as Boolean variables, we can encode the
multiplicative constraint in conjunctive normal form as
\[ (a_k\lor b_k\lor c_k\lor d_k)\land(\lnot a_k\lor\lnot b_k\lor c_k\lor d_k)\land\dotsb\land(\lnot a_k\lor\lnot b_k\lor\lnot c_k\lor\lnot d_k) \]
(that is, all the clauses on the four variables $a_k$, $b_k$, $c_k$, $d_k$
with an even number of negative literals).
We add these clauses for $k=1$, $\dotsc$, $(n-1)/2$ into each SAT instance
generated in each odd order $n$.

\section{Results}\label{sec:results}

We implemented the algorithm described in Section~\ref{sec:sat+casmethod},
including all optimizations,
and ran it in all orders $n\leq70$ divisible by $2$ or $3$.
Step~1 was completed using the script \textsc{nsoks}~\citep{riel2006nsoks} in
the computer algebra system \textsc{Maple~18}. 
Steps~2--4 and the postprocessing were completed
using C++ code which used the library FFTW~3.3.6-pl2 by \cite{frigo2005design}
for computing PSD values.  Step~5 was completed using 
\textsc{MapleSAT}~\citep{DBLP:conf/sat/LiangKPCG17} modified to support a programmatic
interface and also used FFTW for computing PSD values.  
Since FFTW introduces some floating-point errors in the values it returns,
when checking the $\PSD$ values of $A$ we actually ensure that
$\PSD_A(s)>4n+\epsilon$ for some $\epsilon$ which is small but larger than
the accuracy of the discrete Fourier transform used,
e.g., $\epsilon=10^{-2}$.
Our computations were performed on a cluster
of 64\nobreakdash-bit Intel Xeon E5-2683V4 2.1~GHz processors limited to
6~GB of memory and running \mbox{CentOS}~7.

Timings for running our entire algorithm (in hours) in even orders
are given in Table~\ref{tbl:results}, and timings for the running
of the SAT solver alone are given in Table~\ref{tbl:satresults}.
The bottleneck of our method
for large even~$n$ was the matching procedure described in step~4,
which requires enumerating and then sorting a very large number of vectors.
For example, when $n=64$ and $R_A=R_B=8$ there were over $26.6$ billion
vectors added to~$L_{AB}$.
Table~\ref{tbl:results} also includes the number of SAT instances
that we generated in each order, as well as the total number of sets of
Williamson sequences that were found up to equivalence 
(denoted by $\#W_n$).  The counts for $\#W_n$ are not identical to those
given in~\citep{bright2017sat+} because that work did not
use the equivalence operation E5 (alternating negation) but the
results up to order $64$ (the largest order previously solved)
are otherwise identical.

\begin{table}
\begin{center}
\begin{tabular}{c@{\qquad}c@{\qquad}c@{\qquad}c}
$n$        & Time (h)   & \# inst.   & $\#W_n$    \\ \hline
2          & 0.00       & 1          & 1          \\
4          & 0.00       & 1          & 1          \\
6          & 0.00       & 1          & 1          \\
8          & 0.00       & 1          & 1          \\
10         & 0.00       & 2          & 2          \\
12         & 0.00       & 3          & 3          \\
14         & 0.00       & 3          & 5          \\
16         & 0.00       & 5          & 6          \\
18         & 0.00       & 22         & 23         \\
20         & 0.00       & 14         & 17         \\
22         & 0.00       & 22         & 15         \\
24         & 0.00       & 40         & 72         \\
26         & 0.00       & 24         & 26         \\
28         & 0.00       & 78         & 83         \\
30         & 0.00       & 281        & 150        \\
32         & 0.00       & 70         & 152        \\
34         & 0.00       & 214        & 91         \\
36         & 0.00       & 1013       & 477        \\
38         & 0.00       & 360        & 50         \\
40         & 0.01       & 4032       & 1499       \\
42         & 0.02       & 2945       & 301        \\
44         & 0.01       & 1163       & 249        \\
46         & 0.03       & 1538       & 50         \\
48         & 0.09       & 4008       & 9800       \\
50         & 0.45       & 3715       & 275        \\
52         & 0.78       & 4535       & 926        \\
54         & 3.00       & 25798      & 498        \\
56         & 0.98       & 18840      & 40315      \\
58         & 15.97      & 9908       & 73         \\
60         & 27.14      & 256820     & 4083       \\
62         & 64.74      & 19418      & 61         \\
64         & 65.52      & 34974      & 69960      \\
66         & 764.96     & 109566     & 262        \\
68         & 593.77     & 122150     & 1113       \\
70         & 957.96     & 71861      & 98         
\end{tabular}\end{center}
\caption{A summary of the running time in hours, number
of SAT instances used, and number of inequivalent sets of
Williamson sequences generated in each even order $n\leq70$.}\label{tbl:results}
\end{table}

\begin{table}
\begin{center}
\begin{tabular}{ccc}
& \multicolumn{2}{c}{SAT Solving Time (hours)} \\
$n$      & CNF encoding  & Programmatic  \\ \hline
30       & 0.01          & 0.00          \\
32       & 0.01          & 0.00          \\
34       & 0.03          & 0.00          \\
36       & 0.36          & 0.00          \\
38       & 0.12          & 0.00          \\
40       & 2.01          & 0.01          \\
42       & 4.31          & 0.01          \\
44       & 5.59          & 0.00          \\
46       & 8.65          & 0.01          \\
48       & 18.78         & 0.02          \\
50       & 53.12         & 0.02          \\
52       & $-$           & 0.05          \\
54       & $-$           & 0.18          \\
56       & $-$           & 0.18          \\
58       & $-$           & 0.13          \\
60       & $-$           & 9.56          \\
62       & $-$           & 0.45          \\
64       & $-$           & 0.94          \\
66       & $-$           & 5.14          \\
68       & $-$           & 17.18         \\
70       & $-$           & 8.07          
\end{tabular}\end{center}
\caption{The total time spent running \textsc{MapleSAT}
in each even order $30\leq n\leq 70$ using the CNF encoding
and the programmatic encoding. A timeout of 100 hours was used.}\label{tbl:satresults}
\end{table}

Table~\ref{tbl:oddresults} contains the same information
as Table~\ref{tbl:results} except in the odd orders.
The bottleneck for our algorithm in these orders
was the uncompression step (since uncompressing by a factor
of $3$ is more challenging than uncompressing by a factor of $2$),
i.e., solving the SAT instances.
The counts for $\#W_n$ in these cases exactly match
those given by \cite{holzmann2008williamson} up to~$59$,
the largest order they solved.
We found one previously unknown Williamson sequence of
order $63$ using 466,561 $3$\nobreakdash-compressed
quadruples.
We give this Williamson sequence here explicitly, with `\verb|+|'
representing $1$, `\verb|-|' representing $-1$, and
each sequence member on a new line:
{\microtypesetup{activate=false}
\begin{center}
\verb|+----+-+++-+-++++--++--+-----+-++-+-----+--++--++++-+-+++-+----| \\
\verb|+--++-+-+--++++----+-++-++---+----+---++-++-+----++++--+-+-++--| \\
\verb|+--++-+++-+++---++-----+-+---+----+---+-+-----++---+++-+++-++--| \\
\verb|+++++-++----+++-+-++---+-++++-+--+-++++-+---++-+-+++----++-++++|
\end{center}}

\begin{table}
\begin{center}
\begin{tabular}{c@{\qquad}c@{\qquad}c@{\qquad}c}
$n$        & Time (h)   & \# inst.   & $\#W_n$    \\ \hline
3          & 0.00       & 1          & 1          \\
9          & 0.00       & 3          & 3          \\
15         & 0.00       & 8          & 4          \\
21         & 0.00       & 30         & 7          \\
27         & 0.00       & 172        & 6          \\
33         & 0.01       & 364        & 5          \\
39         & 0.05       & 1527       & 1          \\
45         & 1.12       & 15542      & 1          \\
51         & 4.57       & 17403      & 2          \\
57         & 61.26      & 58376      & 1          \\
63         & 1670.95    & 466561     & 2          \\
69         & 8162.50    & 600338     & 1          
\end{tabular}\end{center}
\caption{A summary of the running time in hours, number
of SAT instances used, and number of inequivalent sets of
Williamson sequences generated in each odd order $n$ divisible by $3$ and less than $70$.}\label{tbl:oddresults}
\end{table}

We also used our enumeration of Williamson sequences of order $2n$ for $n\leq35$
along with Theorem~\ref{thm:8will} to explicitly construct 8-Williamson sequences in all
odd orders $n\leq35$.  Table~\ref{tbl:8willresults} contains the counts 
of how many inequivalent sets of 8\nobreakdash-Williamson sequences that can be constructed in this fashion
(this does not count the total number of 8-Williamson
sequences in order $n$, only those that can be constructed via
the construction of Theorem~\ref{thm:8will}).  We explicitly give one example
of an 8-Williamson sequence of order $35$, with `\verb|+|'
representing $1$ and `\verb|-|' representing~$-1$:
{\small\microtypesetup{activate=false}
\begin{center}
\verb|++++-+++--+-+----++----+-+--+++-+++ +++---+-+++-++--+--+--++-+++-+---++| \\
\verb|++-+-+++-+-----++++++-----+-+++-+-+ ++-+--+--++---+-++++-+---++--+--+-+| \\
\verb|++---++-+-+--+--------+--+-+-++---+ ++---++-+-+--+--------+--+-+-++---+| \\
\verb|+--+++-----+---+-++-+---+-----+++-- +---++--++++-++-+--+-++-++++--++---|
\end{center}}\noindent
These sequences can be used to generate a Hadamard matrix of order $8\cdot35=280$;
for details see \cite{kotsireas2006constructions,kotsireas2009hadamard}.

\begin{table}
\begin{center}
\begin{tabular}{cccccccccc}
$n$      & 1 & 3 & 5 & 7 & 9 & 11 & 13 & 15 & 17 \\
$\#8W_n\geq$ & 1 & 1 & 1 & 4 & 13 & 10 & 18 & 129 & 79 \\[-0.75em] \\
$n$      & 19 & 21 & 23 & 25 & 27 & 29 & 31 & 33 & 35 \\
$\#8W_n\geq$ & 43 & 280 & 48 & 257 & 486 & 71 & 58 & 240 & 78 
\end{tabular}
\end{center}
\caption{The number of inequivalent sets of 8-Williamson sequences generated using Theorem~\ref{thm:8will}
in each odd order $n\leq35$ (and therefore a lower bound on $\#8W_n$, the number of inequivalent
sets of 8-Williamson sequences in order $n$).}\label{tbl:8willresults}
\end{table}

Additionally, our code and an explicit enumeration of all the Williamson
sequences and 8-Williamson sequences that we constructed have been made available
on our website \href{https://uwaterloo.ca/mathcheck}{\nolinkurl{uwaterloo.ca/mathcheck}}.

\section{Conclusion and advice}\label{sec:conclusion}

In this paper we have shown the power of the SAT+CAS paradigm (i.e.,
the technique of applying the tools from the fields of satisfiability checking
and symbolic computation)
as well as the power and flexibility of the programmatic SAT approach.
Our focus was applying the SAT+CAS paradigm to the Williamson conjecture
from combinatorial design theory, but we believe the SAT+CAS paradigm
shows promise to be applicable to many other problems and conjectures.
In fact, the SAT+CAS paradigm has recently been used to enumerate
complex Golay pairs~\citep{bright2018enumeration} and good matrices~\citep{bright2019good}.
However, the SAT+CAS paradigm is not something that can
be effortlessly applied to problems or expected to be effective
on all types of problems.
Our experiments in this area allow us to offer some guidance about
the kind of problems in which the SAT+CAS paradigm would work
particularly well.
In particular, \cite{DBLP:phd/basesearch/Bright17}~highlights the
following properties of problems which makes them good candidates
to study using the SAT+CAS paradigm:
\begin{enumerate}
\item \emph{There is an efficient encoding of the problem into a Boolean setting.}
Since the problem has to be translated into a SAT instance or multiple SAT instances
the encoding should ideally be straightforward and easy to compute.  Not only does
this make the process of generating the SAT instances easier and less error-prone
it also means that the SAT solver is executing its search through a domain which
is closer to the original problem. In general, the more convoluted the encoding the less likely
the SAT solver will be able to efficiently search the space.  For example, in our
application we were fortunate to be able to encode $\pm1$ values as Boolean variables.
\item \emph{There is some way of splitting the Boolean formula into multiple instances
using the knowledge from a CAS\@.}  Of course, a SAT instance can always be split into multiple
instances by hard-coding the values of certain variables and then generating instances
which cover all possible assignments of those variables.  However, this strategy is typically
not a good way of splitting the search space. The instances generated in this fashion
tend to have wildly unbalanced difficulties, with some very easy instances and some
very hard instances, limiting the benefits of using many processors to search the space.
Instead, the process of splitting using domain-specific
knowledge allows instances which cannot be ruled out \emph{a priori} to not
even need to be generated because they encode some part of the search space which
can be discarded based on domain-specific knowledge.  For example, in our application
we only needed to generate SAT instances with a few possibilities for the rowsums
of the sequences $A$, $B$, $C$, and~$D$ and could ignore all other possible rowsums.
\item \emph{The search space can be split into a reasonable number of cases.}
One of the disadvantages of using SAT solvers is that it can be difficult to tell how much progress
is being made as the search is progressing.
The process of splitting the search space allows one to get a better estimate of the progress
being made, assuming the difficulty of the instances isn't extremely unbalanced.
In our experience, splitting the search space into instances which can be solved
relatively quickly worked well, assuming the number of instances isn't too large so
that the overhead of calling the SAT solver is small.
This allowed the space to be searched significantly faster (especially
when using multiple processors) than a single instance would have taken to complete.
In our application the order $n=69$
required the most amount of splitting; in this case we split the search
space into 600,338 SAT instances and each instance took an average of $48.8$ seconds to solve.
\item \emph{The SAT solver can learn something about the space as the search is running.}
The efficiency of SAT solvers is in part due to the facts that they learn as the search
progresses.  It can often be difficult for a human to make sense of these facts but they play a vital
role to the SAT solver internally and therefore a problem where the SAT solver can take
advantage of its ability to learn nontrivial clauses is one in which the SAT+CAS paradigm is
well suited for.  
For more sophisticated
lemmas that the SAT solver would be unlikely to learn (because they
rely on domain-specific knowledge) it is useful to learn clauses programmatically via
the programmatic SAT idea~\citep{DBLP:conf/sat/GaneshOSDRS12}.
For example, the timings in Table~\ref{tbl:satresults} show how
important the learned programmatic clauses were to the efficiency of the SAT solver.
\item \emph{There is domain-specific knowledge which can be efficiently given to the SAT solver.}
Domain-specific knowledge was found to be critical to solving instances of the
problems besides those of the smallest sizes.
The instances which were generated
using naive encodings were typically only able to be solved for small sizes
and all significant increases in the size of the problems past that point 
came from the usage of domain-specific knowledge. 
Of course, for the information to be useful to the solver there needs to be an efficient way
for the solver to be given the information; it can be encoded directly in the SAT instances
or generated on-the-fly using programmatic SAT functionality.
For example, in our application we show how to encode
Williamson's product theorem for odd orders $n$ directly into the SAT instance in Section~\ref{sec:optimizations}
and we show how to programmatically encode the PSD test in Section~\ref{sec:progwill}.
\item \emph{The solutions of the problem lie in spaces which cannot be simply enumerated.}
If the search space is highly structured and there exists an efficient search algorithm 
that exploits that structure then using this algorithm directly
is likely to be a better choice. 
A SAT solver could also perform this search
but would probably do so less efficiently; instead,
SAT solvers have a relative advantage when the search space is less structured.
For example, in our application we require searching for sequences
whose compressions are equal to some given sequence and use the PSD test to filter
certain sequences from consideration.  The space is specified by a number of simple 
but ``messy'' constraints and SAT solvers are good at dealing with that kind of complexity.
\end{enumerate}

Perhaps the most surprising result of our work on the Williamson conjecture
is our discovery that there are typically many more Williamson matrices
in even orders than there are in odd orders.  In fact, every odd order $n$
in which a search has been carried out has $\#W_n\leq10$,
while we have shown that every even order $18\leq n\leq 70$ has $\#W_n>10$
and there are some orders which contain thousands of inequivalent
Williamson matrices.
Part of this dichotomy can be explained by Theorem~\ref{thm:willdbl}
which generates Williamson matrices of order $2n$ from Williamson matrices
of odd order $n$.
For example, the two classes of Williamson matrices
of order $10$ can be generated from the single class of Williamson matrices
of order $5$.
However, this still does not fully explain the relative abundance of Williamson
matrices in even orders.  In particular,
it cannot possibly explain why Williamson matrices exist in order $70$
because Williamson matrices of order $35$ do not exist.

Recently \cite{acevedo2019new} discovered
a remarkable construction for Williamson matrices using perfect
quaternionic sequences.  Their construction can be used to explain the existence
of Williamson matrices of order~$70$
and we were able to use their method to construct~$40$
of the~$98$ inequivalent sets of Williamson matrices that we found in order~$70$.

\section*{Acknowledgements}

This work was made possible because of the resources available on the
petascale supercomputer Graham at the University of Waterloo,
managed by Compute Canada and SHARCNET,
the Shared Hierarchical Academic Research Computing Network.
The authors thank Dragomir \DJ okovi\'c 
informing us about the construction of \cite{turyn1970}
using complex Hadamard matrices.
We also thank the reviewers for their very detailed reading of our work
and providing comments that greatly improved this article.

\bibliography{jsc-willsat}

\section*{Appendix: Proofs}

\renewcommand{\thetheorem}{\ref{thm:willprodeven}}
\begin{theorem}
If\/ $A$, $B$, $C$, $D$ are Williamson sequences of even order\/ $n=2m$ then
\[a_ib_ic_id_i=a_{i+m}b_{i+m}c_{i+m}d_{i+m} \qquad \text{for\/ $0\leq i<m$.} \]
\end{theorem}
\begin{proof}
We can equivalently consider members of Williamson sequences to be
elements of the group ring $\Z[C_n]$ where $C_n$ is a cyclic group
of order $n$ with generator $u$.  In such a formulation we have
$X=x_0+x_1u+\dotsb+x_{n-1}u^{n-1}$ and Williamson sequences are
quadruples $(A,B,C,D)$ whose members have $\pm1$ coefficients,
whose coefficients form symmetric sequences of length $n$, and which satisfy
\[ A^2 + B^2 + C^2 + D^2 = 4n . \]
Let $P_X=\sum_{x_i=1} u^i$ (with the sum over $0\leq i<n$) and let $p_X$ denote the number of positive
coefficients in $X$.
As shown in~\cite[14.2.20]{hall1998combinatorial} we have that
$P_A^2 + P_B^2 + P_C^2 + P_D^2$ is equal to
\[ (p_A+p_B+p_C+p_D-n)\sum_{i=0}^{n-1}u^i + n . \tag{1}\label{eq:prodeq1} \]
Furthermore, by the fact that $P_X^2\equiv\sum_{x_i=1} u^{2i}\pmod{2}$,
$P_A^2 + P_B^2 + P_C^2 + P_D^2$ is
congruent to
\[ \sum_{a_i=1} u^{2i} + \sum_{b_i=1} u^{2i} + \sum_{c_i=1} u^{2i} + \sum_{d_i=1} u^{2i} \pmod{2} . \tag{2}\label{eq:prodeq2} \]
Now, if $n$ is even then~\eqref{eq:prodeq1} reduces to
\[ (p_A+p_B+p_C+p_D)\sum_{i=0}^{n-1}u^i \pmod{2} \]
so all coefficients are the same mod $2$.  Since by~\eqref{eq:prodeq2} the
coefficients with odd index are $0$ mod $2$, all coefficients in~\eqref{eq:prodeq1}
and~\eqref{eq:prodeq2} must be $0$ mod $2$.

Note that $u^k=u^{2i}$ has exactly 2 solutions for given even~$k$ with $0\leq k<n$, namely,
$i=k/2$ and $i=(k+n)/2$.  Then~\eqref{eq:prodeq2} can be rewritten as
\[ \sum_{a_{k/2}=1} u^k + \sum_{a_{(k+n)/2}=1} u^k + \dotsb + \sum_{d_{(k+n)/2}=1} u^k \pmod{2} \]
where the sums are over the even $k$ with $0\leq k<n$.
Since each coefficient must be $0$ mod $2$, there must be an even number of $1$s
among the entries $a_{k/2}$, $a_{(k+n)/2}$, $\dotsc$, $d_{(k+n)/2}$ for each even $k$ with $0\leq k<n$, i.e.,
\[ a_{k/2}a_{(k+n)/2}b_{k/2}b_{(k+n)/2}c_{k/2}c_{(k+n)/2}d_{k/2}d_{(k+n)/2} = 1 . \]
The required result is a rearrangement of this and rewriting with
the definition $i=k/2$.
\end{proof}

\renewcommand{\thetheorem}{\ref{cor:willprod}}
\begin{corollary}
If\/ $A'$, $B'$, $C'$, $D'$ are the\/ $2$\nobreakdash-compressions of a set of Williamson sequences
then\/ $A'+B'+C'+D'\equiv[0,\dotsc,0]\pmod{4}$.
\end{corollary}
\begin{proof}
Let $N_+$ and $N_-$ denote the number of $1$s and $-1$s in the eight Williamson sequence entries
$a_i$, $b_i$, $c_i$, $d_i$, $a_{i+m}$, $b_{i+m}$, $c_{i+m}$, and $d_{i+m}$,
where $0\leq i<m$.  We have that $N_++N_-=8$ and that
$N_+-N_-=a'_i+b'_i+c'_i+d'_i$
(the sum of the above eight Williamson sequence entries).
Thus the $i$th entry of $A'+B'+C'+D'$ is
$N_+-N_-=N_+-(8-N_+)=2N_+-8\equiv0\pmod 4$ since
Theorem~\ref{thm:willprodeven} implies that $N_+$ must be even.
\end{proof}

\renewcommand{\thetheorem}{\ref{thm:willdbl}}
\begin{theorem}
Let $A$, $B$, $C$, $D$ be Williamson sequences of odd order $n$.
Then
{\rm\[ A\shuffle B',\, (-A)\shuffle B',\, C\shuffle D',\, (-C)\shuffle D' \]}%
are Williamson sequences of order $2n$.
\end{theorem}
\begin{proof}
The fact that the constructed sequences have $\pm1$ entries are of length $2n$
follows directly from the properties of the three types of operations used to generate them.
The fact that they are symmetric follows from the fact that the
sequences $X$ which appear to the left of $\shuffle$ satisfy $x_k=x_{n-k}$
for $k=1$, $\dotsc$, $n-1$ and the sequences $Y$ which appear to the right
of $\shuffle$ satisfy $y_k=y_{n-k-1}$ for $k=0$, $\dotsc$, $n-1$ which are
exactly the necessary properties for $X\mspace{-1mu}\shuffle\mspace{1mu}Y$ to be symmetric.

Let $L$ be the list containing the constructed sequences of order $2n$.
To show these sequences are Williamson we need to show that
\[ \sum_{X\in L}\PAF_X(s) = 0 \]
for $s=1$, $\dotsc$, $n$.  Using the properties that 
$\PAF_{-X}(s)=\PAF_X(s)$, $\PAF_{X'}(s)=\PAF_X(s)$, and
$\PAF_{X\shuffle Y}(s)=\PAF_X(s/2)+\PAF_Y(s/2)$ when $s$ is even and in this range
we obtain
\[ \sum_{X\in L}\PAF_X(s) = 2\sum_{X=A,B,C,D}\PAF_X(s/2) = 0 \]
since $A$, $B$, $C$, $D$ are Williamson.  When $s$ is odd we have that
\[ \PAF_{(-X)\shuffle Y}(s) = -\PAF_{X\shuffle Y}(s) \]
and using this for $(X,Y)=(A,B')$ and $(C,D')$ derives the desired property.
\end{proof}

\renewcommand{\thetheorem}{\ref{thm:8will}}
\begin{theorem}
Let $A$, $B$, $C$, $D$ be Williamson sequences of order $2n$
with $n$ odd and write
{\rm\[ A = A_1\shuffle A_2',\, B = B_1\shuffle B_2',\, C = C_1\shuffle C_2',\, D = D_1\shuffle D_2' . \]}%
Then $A_1$, $A_2$, $B_1$, $B_2$, $C_1$, $C_2$, $D_1$, $D_2$ are 8-Williamson sequences of order $n$.
\end{theorem}
\begin{proof}
The fact that the constructed sequences are symmetric, have $\pm1$ entries,
and are of order $n$ follows directly from the construction and because
the sequences they are constructed from are symmetric,
have $\pm1$ entries, and are of order $2n$.  Since $A$, $B$, $C$, $D$ are
Williamson we have that
\[ \sum_{X=A,B,C,D}\PAF_X(2s) = 0 \qquad\text{for $s=1$, $\dotsc$, $n-1$} . \]
Using the fact that $\PAF_{X\shuffle Y}(2s)=\PAF_X(s)+\PAF_Y(s)$
and $\PAF_{Y'}(s)=\PAF_Y(s)$ this sum becomes exactly the Williamson property
$\sum_{i=1}^8 \PAF_{X_i}(s) = 0$.
\end{proof}

\end{document}